\documentclass[11pt, letter]{article}
\usepackage[margin=1in]{geometry}
\usepackage{amsmath, amssymb, amsthm, thmtools, amsfonts, bm}
\usepackage{comment}
\usepackage{bbm}
\usepackage{cite}
\usepackage{appendix}
\usepackage{graphicx}
\usepackage{color}
\usepackage{algorithm}
\usepackage{algorithmicx}
\usepackage[noend]{algpseudocode}
\usepackage{epstopdf}
\usepackage{wrapfig}
\usepackage{paralist}
\usepackage[textsize=tiny]{todonotes}

\usepackage{framed}
\usepackage[framemethod=tikz]{mdframed}
\usepackage[bottom]{footmisc}
\usepackage{enumitem}
\setitemize{noitemsep,topsep=3pt,parsep=3pt,partopsep=3pt}
\usepackage[font=small]{caption}
\usepackage{xspace}

\theoremstyle{plain}
\newtheorem{thm}{Theorem}[section]
\newtheorem{cor}[thm]{Corollary}

\newtheorem{lem}[thm]{Lemma}

\newtheorem{Def}[thm]{Definition}
\newtheorem{obs}[thm]{Observation}

\definecolor{darkgreen}{rgb}{0,0.5,0}
\usepackage{hyperref}
\usepackage[capitalize, nameinlink]{cleveref}
\crefname{listing}{Program-code}{Program-codes}  
\crefname{theorem}{Theorem}{Theorems}
\Crefname{lemma}{Lemma}{Lemmas}
\Crefname{observation}{Observation}{Observations}
\Crefname{definition}{Definition}{Definitions}
\Crefname{Def}{Definition}{Definitions}

\algnewcommand\algorithmicswitch{\textbf{switch}}
\algnewcommand\algorithmiccase{\textbf{case}}

\algdef{SE}[SWITCH]{Switch}{EndSwitch}[1]{\algorithmicswitch\ #1\ \algorithmicdo}{\algorithmicend\ \algorithmicswitch}%
\algdef{SE}[CASE]{Case}{EndCase}[1]{\algorithmiccase\ #1}{\algorithmicend\ \algorithmiccase}%
\algtext*{EndSwitch}%
\algtext*{EndCase}%

\newcommand{\eps}{\varepsilon}
\newcommand{\poly}{\operatorname{\text{{\rm poly}}}}

\renewcommand{\paragraph}[1]{\vspace{0.15cm}\noindent {\bf #1}:}

\newcommand{\FullOrShort}{short}

\ifthenelse{\equal{\FullOrShort}{full}}{
	
	\newcommand{\fullOnly}[1]{#1}
	\newcommand{\shortOnly}[1]{}
	
}{

\newcommand{\fullOnly}[1]{}
\newcommand{\IncludePictures}[1]{}

}

\usepackage{authblk}

\title{Simplified and Space-Optimal Semi-Streaming for ($2+\epsilon$)-Approximate Matching}

\author[1]{Mohsen Ghaffari}
\author[2]{David Wajc}
\affil[1]{ETH Zurich}
\affil[2]{Carnegie Mellon University} 

\begin{document}

\maketitle
\begin{abstract}
	In a recent breakthrough, Paz and Schwartzman (SODA'17) presented a single-pass ($2+\eps$)-approximation algorithm for the maximum weight matching problem in the semi-streaming model. Their algorithm uses $O(n\log^2 n)$ bits of space, for any constant $\eps>0$. 
	
	We present two simplified and more intuitive analyses, for essentially the same algorithm, which also improve the space complexity to the optimal bound of $O(n\log n)$ bits --- this is optimal as the output matching requires $\Omega(n\log n)$ bits. Our analyses rely on a simple use of the primal dual method and a simple accounting method.
\end{abstract}

\section{Introduction and Related Work}
The maximum weight matching (MWM) problem is a classical optimization problem, with diverse applications, which has been studied extensively since the 1965 work of Edmonds \cite{edmonds1965paths}. Naturally, this problem has received significant attention also in the \emph{semi-streaming} model. This is a modern model of computation, introduced by Feigenbaum et al.\cite{feigenbaum2005graph}, which is motivated by the need for processing massive graphs whose edge set cannot be stored in memory. In this model, roughly speaking, the edges of the graph arrive in a stream, and the algorithm should process this stream and eventually output its solution --- a matching in our case --- while using a small memory at all times. Ideally, this memory size should be close to what is needed for storing just the output. More formally, the setting of the MWM problem in the semi-streaming matching is as follows.

\medskip
\paragraph{MWM in the Semi-Streaming Model} Let $G=(V, E, w)$ be a simple graph with non-negative edge weights $w:E\rightarrow \mathbb{R}_{>0}$ (for notational simplicity, we let $w_e=w(e)$). Let $n=|V|$, $m=|E|$. We assume that the edge weights are normalized so that the minimum edge weight is $1$ and we use $W=\max_e w_e$ to denote the maximum edge weight. In the semi-streaming model, the input graph $G$ is provided as a stream of edges. In each iteration, the algorithm receives an edge from the stream and processes it. The algorithm has a memory much smaller than $m$ and thus it cannot store the whole graph. The amount of the memory that the algorithm uses is called its \emph{space complexity} and we wish to keep it as small as possible. The objective in the semi-streaming \emph{maximum weight matching} (MWM) problem is that, at the end of the stream, the algorithm outputs a matching whose weight is close to the weight of the maximum weight matching, denoted by $M^* = \arg\max_{\text{matching } M} w(M)$, where $w(M) = \sum_{e\in M} w_e$ is the weight of matching $M$.

\medskip
\paragraph{State of the Art} There has been a sequence of successively improved approximation algorithms for MWM in the semi-streaming model. Feigenbaum et al.~gave a $6$ approximation\cite{feigenbaum2005graph}, McGregor gave a $5.828$ approximation\cite{mcgregor2005finding} (and a $2+\eps$ approximation, but using $O(1/\eps^3)$ passes on the input), Epstein et al. gave a $4.911+\eps$ approximation\cite{epstein2011improved}, Crouch and Stubbs improved this bound to $4+\eps$ \cite{crouch2014improved}. However, these approximations remained far from the more natural and familiar $2$ approximation which the sequential greedy method provides. 

In a recent breakthrough, Paz and Schwartzman\cite{PS17-StreamingMatching} presented a truly simple algorithm that achieves an approximation of $2+\eps$ for any constant $\eps>0$, using $O(n\log^2 n)$ bits of space. More concretely, the algorithm maintains $O(n\log n)$ edges, while working through the stream, and at the end, it computes a matching using these maintained edges.

\medskip
\paragraph{Our Contribution} We present alternative analyses for (a slight variant of) the algorithm of Paz and Schwartzman, which have two advantages: (1) they imply that keeping merely $O(n)$ edges suffices, and thus improve the space complexity to $O(n\log n)$ bits, which is optimal, (2) they provide a more intuitive and also more flexible accounting of the approximation. 

Concretely, Paz and Schwartzman\cite{PS17-StreamingMatching} used an extension of \emph{the Local Ratio theorem}\cite{bar2001unified, bar2004local} to analyze their algorithm. We instead present two different simpler arguments, which are more flexible and allow us to improve the space-complexity to obtain the optimal space bound. Our analyses rely on the primal dual method and an intuitive \emph{blaming} argument. The main appeal of the primal dual method is in the simple explanation it provides for the extension of the local ratio technique in \cite{PS17-StreamingMatching} using little more than LP duality.
The main appeal of the blaming argument is that it uses completely elementary arguments and can be taught in any algorithms class.

\medskip
\paragraph{Roadmap} In \Cref{sec:review}, as a warm up, we review (a simple version of) the algorithm of Paz and Schwartzman. In \Cref{sec:Analysis}, we present simpler and more intuitive styles of analysis for this algorithm. In \Cref{sec:improve}, we present the improved algorithm that achieves a space complexity of $O(n\log n)$, the analysis of which relies crucially on the flexibility of the analyses presented in \Cref{sec:Analysis}.

\section{Reviewing the Algorithm of Paz and Schwartzman}
\label{sec:review}
\subsection{The Basic Algorithm}
The starting point in the approach of Paz and Schwartzman\cite{PS17-StreamingMatching} is the following basic yet elegant algorithm, implicit in \cite[Section 3]{bar2001unified}. For the sake of explanation, consider a sequential model of computation. We later discuss the adaptation to the streaming model.

\begin{center}
	\begin{minipage}{0.95\linewidth}
		\vspace{-2pt}
		\begin{mdframed}[hidealllines=true, backgroundcolor=gray!20]
			\paragraph{Basic Algorithm} Repeatedly select an edge $e$ with positive weight; reduce its weight from itself and its neighboring edges; push $e$ into a stack and continue to the next edge, so long as edges with positive weight remain. At the end; unwind the stack and add the edges greedily to the matching.
		\end{mdframed}
	\end{minipage}
\end{center}

In \Cref{sec:Analysis} we will see that this simple algorithm is $2$-approximate.

\paragraph{Implementing the Basic Algorithm in the Semi-Streaming Model} To implement this while streaming, we just need to remember a parameter $\phi_v$ for each node $v$. This parameter is the total sum of the weight already reduced from the edges incident on vertex $v$, due to edges incident on $v$ that were processed and put in the stack before.  We assume for now that the space requirement of storing these $n$ numbers is only $O(n \log n)$ bits (say, due to the edge weights being at most $\poly(n)$) and discuss the space requirement of these $\phi_v$ values at the end of the paper.

\subsection{The Algorithm with Exponentially Increasing Weights}\label{sec:expAlg}
The space complexity of the basic algorithm can become $\Theta(n^2)$, as it may end up pushing $\Theta(n^2)$ edges into the stack. This brings us to the clever idea of Paz and Schwartzman \cite{PS17-StreamingMatching}, which reduces the space complexity to the equivalent of keeping $O(n\log W/\eps)$ edges, where $W$ is the normalized maximum edge weight, while still providing a $(2+\eps)$-approximation. 

The idea is to ensure that the edges incident on each node $v$ that get pushed into the stack have exponentially increasing weights, by factors of $(1+\eps)$. Thus, at most $O(\log_{1+\eps} W) = O((\log W)/\eps)$ edges per node are added to the stack, and the overall number of edges in the stack is at most $O((n\log W)/\eps)$. 

To attain this exponential growth, the method is as follows: When reducing the weight of an edge $e$ from each neighboring edge $e'$, we will decide between deducting either $w_e$ or $(1+\eps) w_e$ from the weight $w_{e'}$. In general, this can be any arbitrary decision.  This arbitrary choice may seem mysterious, but as we shall see in \Cref{sec:Analysis}, the particular choice is rather natural when considered in terms of LP duality. In the streaming model, this decision is done when we first see $e'=\{u,v\}$ in the stream, as follows.
\begin{itemize}
	\item If $w_{e'} < (1+\eps)\cdot (\phi_u+\phi_{v})$ --- i.e., if $e'$ has less than ($1+\eps$) times of the total weight of the stacked up edges incident on $u$ or $v$ --- then we deduct $(1+\eps) w_e$ from $w_{e'}$ for each stacked up edge $e$ incident on $u$ or $v$. Hence, we effectively  reduce the weight of $w_{e'}$ by $(1+\eps)\cdot (\phi_u+\phi_{v})$. This implies that we get left with an edge $e'$ of negative weight, which can be ignored without putting it in the stack. 
	\item Otherwise, if $w_{e'} \geq  (1+\eps)\cdot (\phi_u+\phi_{v})$, we deduct only $w_e$ from $w_{e'}$, for each edge $e$ incident on $u$ or $v$ that is already in the stack. Thus, in total, we deduct only $(\phi_u+\phi_{v})$ weight from $w_{e'}$ for the previously stacked edges. Thus, now we have an edge $e'$ whose leftover weight is $w'_{e'} \geq (1+\eps)\cdot (\phi_u+\phi_{v})- (\phi_u+\phi_{v}) = \eps\cdot  (\phi_u+\phi_{v})$. Then, we add $e'$ to the stack, and thus $\phi_u$ and $\phi_{v}$ increase by $w'_{e'}$. Therefore, both $\phi_u$ and $\phi_{v}$ increase by at least a $(1+\eps)$ multiplicative factor. 
\end{itemize}
The concrete algorithm that formalizes the above scheme is presented in \Cref{alg:PSAlg}.  

\begin{algorithm}[h]
	\caption{The Paz-Schwartzman Algorithm \cite{PS17-StreamingMatching} with Exponentially Increasing Weights}
	\label[algorithm]{alg:PSAlg}
	\begin{algorithmic}[1]
		\Statex
		\State $S\gets \textit{emptystack}$
		\State $\forall v\in V: \phi_v=0$
		\For{$e=\{u, v\} \in E$} 
		\If{$w_e < (1+\eps)\cdot (\phi_u+\phi_{v})$}\label{line:covered}
		\State continue \Comment{skip to the next edge}
		\EndIf
		\State $w'_e \gets w_e - (\phi_u+\phi_{v})$
		\State $\phi_u \gets \phi_u + w'_e$ \label{line:phi_u}
		\State $\phi_{v} \gets \phi_{v} + w'_e$ \label{line:phi_v}
		\State $S.push(e)$
		\EndFor
		\Statex \vspace{-0.25cm}
		\State $M \gets \emptyset$
		\While{$S\neq \emptyset$} 
		\State $e\gets S.pop()$
		\If{$M\cap N(e) = \emptyset$}
		$M\gets M\cup \{e\}$
		\EndIf
		\EndWhile
		\State\Return $M$
	\end{algorithmic}
\end{algorithm}

\begin{obs}\label[observation]{obs:exponentialIncrease} When an edge $e=\{u, v\}$ is added to the stack, the value of $\phi_u$ increases by a $1+\eps$ factor.
\end{obs}

The above observation also implies that the edges incident on each node that are added to the stack have an exponential growth in weight, which directly implies that the total number of edges pushed to the stack cannot exceed $O(n\log_{1+\eps} W) = O((n\log W)/\eps)$. 
In \Cref{sec:Analysis}, we prove that this algorithm is $2(1+\eps)$-approximate.

\section{Simplified Analyses}\label{sec:Analysis}
In this section we give two simple analyses of \Cref{alg:PSAlg}, proving it yields a $2(1+\eps)$ approximation of the MWM. Note that the basic algorithm above can be seen as \Cref{alg:PSAlg} run with $\eps=0$, and so we obtain that the basic algorithm is a $2$-approximate algorithm. 
\begin{lem}\label[lemma]{lem:expAlg} The matching $M$ returned by \Cref{alg:PSAlg} is a $2(1+\eps)$ approximation of the maximum weight matching.
\end{lem}

\subsection{Duality-based Analysis}\label{sec:dualitySimple}
	The first observation our duality-based proofs rely on is that $\vec{\phi_v}$ forms a (nearly) feasible solution to the dual of the LP relaxation of the MWM problem, in \Cref{fig:mwm-lp}. 
	Indeed, this fact is not accidental, and it is the underlying reason for the choice of when to decrease weights of an edge neighboring a processed edge $e$ by $(1+\eps)w_e$ or by $w_e$.

	\begin{figure}[h]
		\begin{small}
			\begin{center}
				\begin{tabular}{rl|rl}
					\multicolumn{2}{c|}{Primal} & \multicolumn{2}{c}{Dual} \\ \hline  
					maximize & $\sum_{e\in E} w_{e}\cdot x_{e}$ &
					minimize & $\sum_{v\in V}{y_v}$ \\
					subject to: & & subject to: & \\
					
					$\forall v\in V$: &  $\sum_{e\ni v}x_{e} \leq 1$ &
					$\forall \{u,v\}\in E$: & $y_u + y_v \geq w_{\{u,v\}}$ \\
					
					$\forall e\in E$: & $x_{e} \geq 0$ & $\forall v\in V$: & $y_v \geq0$
				\end{tabular}
			\end{center}
		\end{small}
		\vspace{-0.5cm}
		\caption{The LP relaxation of MWM and its dual}
		\label{fig:mwm-lp}
	\end{figure} 
	
	\begin{obs}\label[observation]{obs:approx-dual}
		The variables $\{\phi_v\}_{v\in V}$ scaled up by $1+\eps$ form a feasible dual solution.
	\end{obs}
	\begin{proof}
		Each edge $e=\{u,v\}\in E$ has it its dual constraint satisfied by $(1+\eps)\cdot \vec{\phi}$ after inspection; i.e., $w_e\leq (1+\eps)\cdot (\phi_u + \phi_{v})$. This constraint is  either satisfied before $e$ is inspected, or after performing lines \ref{line:phi_u} and \ref{line:phi_v}, following which the  new and old values of $\phi_u$ and $\phi_{v}$ satisfy $w_e \leq w_e + w'_e = (\phi^{old}_u+\phi^{old}_{v}+w'_e)+w'_e = \phi^{new}_u + \phi^{new}_{v} \leq (1+\eps)\cdot (\phi^{new}_u + \phi^{new}_{v})$.  As the $\phi_v$ variables never decrease, each dual constraint is satisfied by $(1+\eps)\cdot \vec{\phi}$ in the end.
	\end{proof}
	
	By LP duality, the above implies the following upper bound on the optimal matching.
	\begin{cor}\label[corollary]{cor:phi_ub_MWM}
		The weight of any matching $M^*$ satisfies $$w(M^*) \leq OPT(LP) \leq (1+\eps)\cdot \sum_v \phi_v.$$
	\end{cor}
	
	Let $\Delta\phi^e$ be the change to $\sum_{v\in V} \phi_v$ in Lines \ref{line:phi_u} and \ref{line:phi_v} of the algorithm during the inspection of edge $e$. Note that if the algorithm does not reach these lines when inspecting edge $e$ (due to the test in Line \ref{line:covered}), then we have $\Delta \phi^e=0$. By definition, at the time that the algorithm terminates, $\sum_v \phi_v = \sum_e \Delta \phi^e$. The following lemma relates the weight of an edge $e$ to $\Delta\phi^{e'}$ of edges $e'$ incident on a common vertex with $e$ (including $e$) inspected no later than $e$.
	
	\begin{lem}\label[lemma]{lem:weight-vs-dual-change}
		For each edge $e=\{u,v\}\in E$ added to the stack $S$, if we denote its preceding neighboring edges by $\mathcal{P}(e)=\{e' \mid e'\cap e \neq \emptyset,\, e' \mbox{ inspected no later than } e\}$, then $$w_e \geq \frac{1}{2}\cdot \sum_{e'\in \mathcal{P}(e)} \Delta\phi^{e'} = \sum_{e'\in \mathcal{P}(e)} w'_{e'}.$$
	\end{lem}
	\begin{proof}\vspace{-0.1cm}
		First, we note that $\Delta\phi^e = 2w'_e$, due to Lines \ref{line:phi_u} and \ref{line:phi_v}, or put otherwise $w'_e=\frac{1}{2}\cdot \Delta\phi^e$. On the other hand, if we denote the values of $\phi_u$ and $\phi_v$ before 
		the inspection of $e$ by $\phi_u^e$ and $\phi_v^e$, respectively, we have that $\phi_u^e + \phi_v^e \geq \frac{1}{2}\cdot  \sum_{e'\in \mathcal{P}(e)\setminus \{e\}} \Delta\phi^{e'}$ (this inequality is an equality unless $G$ contains edges parallel to $e$). Consequently, we have that indeed 
		$$
		w_e = w'_e + (\phi_u^e + \phi_v^e) \geq \frac{1}{2}\cdot \sum_{e'\in \mathcal{P}(e)} \Delta \phi^{e'} = \sum_{e'\in \mathcal{P}(e)} w'_{e'}.\qedhere
		$$
	\end{proof}

	\begin{proof}[Proof of \Cref{lem:expAlg}] \vspace{-0.05cm} By the algorithm's definition, every edge ever added to $S$ and not taken into $M$ has a previously-inspected neighboring edge taken into $M$. So, by \Cref{lem:weight-vs-dual-change} and \Cref{cor:phi_ub_MWM} we have $w(M)=\sum_{e \in M}w_e \geq \frac{1}{2}\sum_e \Delta\phi^e = \frac{1}{2}\sum_v \phi_v \geq \frac{1}{2(1+\eps)}\cdot w(M^*)$. In other words, the matching $M$ output by \Cref{alg:PSAlg} a $2(1+\eps)$-approximate MWM, as claimed. 
\end{proof}

\subsection{Blaming-based Analysis}\label{sec:blamingSimple}

\begin{proof}[Proof of \Cref{lem:expAlg}] 
Fix some matching $M'$. We prove that $\sum_{e'\in M'} w(e') \leq 2(1+\eps) \sum_{e\in M} w(e)$. We can think of the left hand side summation as being made of $w(e')$ units of ``\emph{blame}'' on each edge $e'\in M'$.	We think of ``blame'' as a mass of matter which we can move around while preserving its total amount. We show that there exists a way of transferring this blame onto the edges of $M$, in a manner that preserves the total blame, and such that at the end, each edge $e\in M$ receives at most $2(1+\eps) w(e)$ blame. Since the total blame is always preserved, this directly implies that $\sum_{e'\in M'} w(e') \leq 2(1+\eps)\sum_{e\in M} w(e)$. In a sense, we are \emph{blaming} the fact that we did not take edges of $M'$ on the edges of $M$. Taking $M'=M^*$, this implies that the matching $M$ output by \Cref{alg:PSAlg} is a $2(1+\eps)$ approximation of the MWM.
	
	
	Consider the time that we remove an edge $e_{1}=\{u,v\}$ and put $e_1$ in the stack. We may deduct either $w_{e_1}$ or $(1+\eps) w_{e}$ from each of the edges incident on $u$ or $v$. That is, for each incident edge $e'$, the leftover weight is either $w_{e'}-w_{e_1}$ or $w_{e'}-(1+\eps)w_{e_1}$. Let us use $w_{r}(e')$ to denote this leftover weight of $e'$. We then get left with a new graph $G_r$ with updated weights $w_{r}$. On the rewind of the stack, when processing $e_1$, we have two possibilities: 
	\medskip
	
	\begin{minipage}{.9\textwidth}
		\begin{itemize} \item[\textbf{(A)}] the matching $\mathcal{M}_r$ computed on $G_r$ has an edge $e''$ incident on $v$ or $v'$, or 
			\item[\textbf{(B)}] the matching $\mathcal{M}_r$ has no edge incident on $u$ or $v$, thus we will add $e_1$ to the matching of $G$. 
		\end{itemize}
	\end{minipage}
	\medskip
	
	\noindent Suppose by induction that on $G_r$, there is a way of transferring a blame of $w_{r}(e')$ units on each edge $e'\in M'$ onto the edges of $\mathcal{M}_r$ such that each edge $e\in \mathcal{M}_r$ receives at most $2(1+\eps)w_{r}(e)$ blame. The matching $M'$ contains at most two edges $e_2$ and $e_3$ incident on $u$ or $v$. By the inductive assumption, for each $e_{i}\in \{e_2, e_3\}$, we have already found room for placing the $w_{r}(e_i)\geq w_{e_i}-(1+\eps)w_{e_1}$ part of the blame on edges of $\mathcal{M}_r$. Since for these two edges we are trying to transfer $w_{e_2}$ and $w_{e_3}$ units of blame onto the computed matching, what remains to be done is to find room for at most $\big(w_{e_2}-w_{r}(e_2)\big) + \big(w_{e_3}-w_{r}(e_3)\big) \leq 2 (1+\eps) w_{e_1}$ more blame. In case (A), edge $e'' \in \mathcal{M}_r$ has at most $2(1+\eps)w_{r}(e'')\leq 2(1+\eps)(w_{e''}-w_{e_1})$ blame on it at the moment. But in $G$, edge $e''$ has room for $2(1+\eps)w_{e''}$ units of blame. Hence, in this case, we have room for that $2(1+\eps) w_{e_1}$ extra blame to be placed on $e''$. In case (B), the edge $e_1$ is a fresh addition to the matching with zero blame on it at the moment. Thus, we again have room for placing the $2(1+\eps)w_{e_1}$ extra blame, this time on $e_1$, while ensuring that each edge receives a blame of at most $2(1+\eps)$ times its weight. 
\end{proof} 

\section{Improved Algorithm}
\label{sec:improve}
The $(2+\epsilon)$-approximate algorithm of the previous section stores $O(n\log W)$ edges for any constant $\epsilon>0$. To improve the space complexity, we would like to keep only $O(n)$ edges, which is asymptotically the bare minimum necessary for keeping just a matching. For that purpose, we will limit the number of edges incident on each vertex $v$ that are in the stack to a constant $\beta=\frac{6\log 1/\eps}{\eps}+1$. When there are more edges, we will take out the earliest one and remove it from the stack. This will be easy to implement using a queue $Q(v)$ for each of vertex $v$, where we keep the length of the $Q(v)$ capped at $\beta$. The pseudo-code is presented in \Cref{alg:OptSpaceAlg}. We will prove in the following that this cannot hurt the approximation factor more than just increasing the parameter $\eps$ by a constant factor.

\begin{algorithm}[h]
	\caption{The Space-Optimal Algorithm}
	\label[algorithm]{alg:OptSpaceAlg}
	\begin{algorithmic}[1]
		\Statex
		\State $S\gets \textit{empty stack}$
		\State $\forall v\in V: Q(v) \gets \textit{empty queue}$
		\State $\forall v\in V: \phi_v=0$
		\For{$e=\{u, v\} \in E$} 
		\If{$w_e < (1+\eps)\cdot (\phi_u+\phi_{v})$}
		\State continue \Comment{skip to the next edge}
		\EndIf
		\State $w'_e \gets w_e - (\phi_u+\phi_{v})$
		\State $\phi_u \gets \phi_u + w'_e$
		\State $\phi_{v} \gets \phi_{v} + w'_e$
		\State $S.push(e)$
		\For{$x\in \{u, v\}$} \label{line:big-queue-start}
		\State $Q(x)$.enqueue($e$)
		\If{$|Q(x)| > \beta = \frac{3\log(1/\eps)}{\eps}+1$} 
		\State $e'\gets Q(x).dequeue()$
		\State remove $e'$ from the stack $S$
		\EndIf \label{line:big-queue-end}
		\EndFor
		\EndFor
		\Statex \vspace{-0.25cm}
		\State $M \gets \emptyset$
		\While{$S\neq \emptyset$} 
		\State $e\gets S.pop()$
		\If{$M\cap N(e) = \emptyset$}
		$M\gets M\cup \{e\}$
		\EndIf
		\EndWhile
		\State\Return $M$
	\end{algorithmic}
\end{algorithm}

\paragraph{Remark} Paz and Schwartzman\cite{PS17-StreamingMatching} used a similar algorithmic idea to keep only $O(n\log n)$ edges in total, instead of $O(n\log W)$ edges.
\footnote{\label{unbounded-W}We note that keeping $O(n\log n)$ edges can be done in a much simpler way, by remembering the maximum edge weight $W_{max}$ observed so far in the stream and discarding all edges of weight below $\delta= \eps W_{\max}/(2(1+\eps)n^2)$. This 
	effectively narrows the range of weights that \Cref{alg:PSAlg} sees to $W'=O(n^2/\eps)$, making its related space complexity $O(n\log n)$. 
	On the other hand, ignoring all edges of weight below $\delta\leq \eps W_{\max}/n^2$ can decrease the MWM, $w(M^*)$, by at most $n^2\delta \leq \eps W_{\max}\leq \eps w(M^*)$; that is, a $(1-\eps)$ multiplicative term. 
	Moreover, for each vertex $v$ the edges of weight at most $\delta$ can increase $\phi_v$ by at most $n\delta = \eps W_{\max}/(2(1+\eps)n)$, thus decreasing the effective weight of edges of weight at least $\delta$ by no more than $(1+\eps)(\phi_u+\phi_{v}) \leq \eps W_{\max}/n\leq \eps w(M^*)/n$. The weight of the maximum weight matching $M^*$ under this new weight function is therefore at most $(1-\eps)$ smaller than under $w$, so running \Cref{alg:PSAlg} on these weights yields a $2(1+O(\eps))$-approximate solution to the MWM under $w$.	
} 
To be precise, they keep $\gamma=\Theta(\log n/\eps)$ edges per node. Unfortunately, this also leads to quite a bit of complications in their analysis, as they need to adapt the local ratio theorem\cite{bar2001unified, bar2004local} to handle the loss. In a rough sense, their argument was that, per step, the process of limiting the queue size to $\gamma$ creates a loss of $(1-\exp(-\gamma))$ factor in the approximation, in the accountings of the local ratio theorem. Thus, over all the $O(n^2)$ edges in the stream, the loss is $(1-\exp(-\gamma))^{O(n^2)}$. The fact that $m$ could be $\Omega(n^2)$ is why they had to set $\gamma=\Theta(\log n)$ to make this loss negligible.

Handling the loss caused by this queue-limitation is much simpler with the simple arguments that we used in \Cref{sec:Analysis}. Furthermore, our analyses will allow us to curtail the per-node queue size to $\beta=O(1)$, while keeping the loss negligible.
We now address the loss due to queue size limitation in Lines \ref{line:big-queue-start}-\ref{line:big-queue-end}, starting with the following observation.

\begin{obs}\label[observation]{obs:DroppingWeights} Suppose that an edge $e=\{u, v\}$ in the stack gets removed from the stack because another edge $e'=\{u, v'\}$ was pushed to the stack later and made the size of the queue $Q(v)$ exceed $\beta$. Then, we say $e'$ \underline{evicted} $e$. In that case, $w'_{e'}\geq w'_e/\eps$ if $\eps\leq 1$. 
\end{obs}
\begin{proof} Since $e$ was evicted by $e'$, there must have been $\beta-1$ edges incident on $u$ that arrived after $e$ (following which $\phi_u\geq w'_e$) but before $e'$ which were pushed into the stack. Hence, because of \Cref{obs:exponentialIncrease}, we have that before the arrival of $e'$, $\phi_u \geq (1+\eps)^{\beta-1} w'_{e} \geq w'_{e}/\eps^2$ (the last inequality holds because $\beta - 1 =  \frac{3\log(1/\eps)}{\eps}\geq \frac{2\log(1/\eps)}{\log(1+\eps)}$ for all $\eps\in (0,1]$).  But since $e'$ was added to the stack, we know that $w'_{e'} \geq \eps (\phi_u + \phi_v') \geq \eps \phi_u \geq w'_{e}/\eps$.
\end{proof}

The following recursive definition will prove useful when bounding the loss due to eviction of edges from the queue. 

\begin{Def}\label{def:discards}
	We say an edge $e'$ which was inserted into $S$ was \underline{discarded} if it was later removed from $S$,
	and say the edge was \underline{kept} otherwise.
	We say a discarded edge $e'$ was \underline{discarded by} a kept edge $e$ if $e'$ was evicted by $e$ or if $e'$ was evicted by an edge $e''$ which was itself later discarded by $e$. That is, there is a sequence of evictions which starts with $e'$ and ends in edge $e$ where in this sequence, each edge is evicted by the next edge in the chain. We denote the set of edges discarded by $e$ by $\mathcal{D}(e)$.
\end{Def}

We now bound the leftover weights of edges discarded by a given edge $e$.
\begin{lem}\label[lemma]{lem:weight-discarded}
	For all $\eps\leq 1/4$, for each edge $e\in E$, we have $w'_e\geq \frac{1}{4\eps}\cdot \sum_{e'\in \mathcal{D}(e)} w'_{e'}$.
\end{lem}
%

\begin{proof} 
	The set $\mathcal{D}(e)$ contains at most two edges evicted by $e$ -- one for each endpoint of $e$. 
	By \Cref{obs:DroppingWeights}, we know that every such edge $e'$ evicted by $e$ satisfies $w'_{e}\geq w'_{e'}/\eps$.
	Similarly, any edge $e'$ evicted by $e$ can evict at most two edges in $\mathcal{D}(e)$, where each such edge $e''$ satisfies $w'_{e''}\leq \eps\cdot w'_{e'}\leq \eps^2 \cdot w'_{e}$, and so on, by induction. Generally, the edges of $\mathcal{D}(e)$ can be partitioned into sets of at most $2^i$ edges each for all $i\in \mathbb{N}$, where edges $e'$ in the $i$-th set have $w'_{e'}\leq \eps^{i}\cdot w'_e$. Summing over all these sets, we find that indeed, as $\eps\leq 1/4$, 
	$$
	\sum_{e'\in \mathcal{D}(e)} w'_{e'} \leq \sum_{i=1}^\infty (2\eps)^i\cdot w'_e \leq 2\eps\cdot \sum_{i=0}^\infty 2^{-i}\cdot w'_e  \leq 4\eps \cdot w'_e. \qedhere
	$$
\end{proof}

Combining the simple arguments of \Cref{sec:Analysis} and \Cref{lem:weight-discarded} we obtain the following.
\begin{thm}\label[theorem]{thm:improvedAlg}
	For all $\eps\leq \frac{1}{4}$, 
	\Cref{alg:OptSpaceAlg} is $2(1+6\eps)$-approximate.
\end{thm}

In the following subsections we provide two simple proofs of this theorem.
\subsection{Duality-based Analysis}
We observe that \Cref{obs:approx-dual} (and consequently \Cref{cor:phi_ub_MWM}) as well as \Cref{lem:weight-vs-dual-change} hold for \Cref{alg:OptSpaceAlg}, just as they did for \Cref{alg:PSAlg}, as their proofs are unaffected by limiting of the queue sizes in Lines \ref{line:big-queue-start}-\ref{line:big-queue-end}, which is the only difference between these algorithms.

\begin{proof}[Proof of \Cref{thm:improvedAlg}]
	By Lemma \ref{lem:weight-vs-dual-change}, for each edge $e\in M$, we have
	$w_e \geq  \sum_{e'\in \mathcal{P}(e)} w'_{e'}.$
	On the other hand,  by Lemma \ref{lem:weight-discarded}, we have that 
	$4{\eps}\cdot w_e\geq 4{\eps}\cdot w'_e \geq  \sum_{e'\in \mathcal{D}(e)} w'_{e'}.$
	Therefore,		 
	$$
	w_e \cdot (1+4\eps) \geq \sum_{e'\in \mathcal{P}(e)} \left( w'_{e'} +  \sum_{e''\in \mathcal{D}(e')} w'_{e''} \right).$$
	
	But each kept edge $e'$ not added to $M$ is due to a kept edge $e$ which was added to $M$; that is, some $e$ such that $e'\in \mathcal{P}(e)$. Likewise, each discarded edge is discarded due to some kept edge. Consequently, the right hand side of the above expression summed over all edges of the output matching $M$ is precisely $\sum_{e'\in E} w'_{e'} = \frac{1}{2}\cdot \sum_{e'\in E} \Delta\phi^{e'} = \frac{1}{2}\cdot \sum_{v\in V} \phi_v$, which by Corollary \ref{cor:phi_ub_MWM} yields
	$$
	\sum_{e\in M} w_e \geq \frac{1}{2(1+4\eps)} \cdot \sum_{v\in V} \phi_v \geq \frac{1}{2(1+4\eps)(1+\eps)}\cdot w(M^*) \geq \frac{1}{2(1+6\eps)}\cdot w(M^*).\qedhere 
	$$
\end{proof}

\subsection{Blaming-based Analysis}

Again, we will continue with our style of \emph{blaming} as in \Cref{sec:blamingSimple}. We now blame the weight of any arbitrary matching $M'$ on the computed matching $M$, in a manner that each edge in $M$ receives at most $2(1+6\eps)$ factor of its weight as blame. 

\begin{proof}[Proof of \Cref{thm:improvedAlg}]
The proof is similar to \Cref{lem:expAlg} modulo one change: The only difference in the algorithm in comparison with the one analyzed in \Cref{lem:expAlg} is that now some edges have been evicted from the stack, because the queue size of one of their endpoints grew larger than $\beta$. Consider one edge $e_1=\{ u, v\}$ that was added to the stack, but later evicted from the stack. Now, the way that we find room for the blame of $e_1$ will be different. Let $e^*$ be the kept edge such that $e_1$ was discarded by $e^*$, i.e., $e_1\in \mathcal{D}(e^*)$. We find room for the  $2(1+\eps) w'_{e_1}$ units of blame of $e_1$ later when processing $e^*$. At that time, if we were to follow the proof of \Cref{lem:expAlg}, we would have been looking for $2(1+\eps) w'_{e^*}$ room for placing the blame of the edge $e^*$ on the computed matching. To prove \Cref{thm:improvedAlg}, when processing $e^*$, we will instead look for $2(1+6\eps) w'_{e^*}$ room for blame. This extra $10\eps w'_{e^*}$ units of blame is enough to account for the blame of all edges in $\mathcal{D}(e^*)$, i.e., all those edges similar to $e_1$ who were discarded by the kept edge $e^*$. The extra $10\eps w'_{e^*}$ units of blame suffices because the total blame for $\mathcal{D}(e^*)$ is at most $\sum_{e''\in \mathcal{D}(e^*)} 2(1+\eps) w'_{e''} \leq 4\eps 2(1+\eps) w'_{e^*} \leq 10\eps w'_{e^*}$. Here, the first inequality follows from \Cref{lem:weight-discarded}. Even though now when processing $e^*$ we are looking for $2(1+6\eps) w'_{e^*}$ room for blame, we can do that now because in this proof each edge has room for blame equal to a $2(1+6\eps)$ factor of its weight, instead of a $2(1+\eps)$ factor of its weight as in \Cref{lem:expAlg}. The rest of the proof is identical to \Cref{lem:expAlg}.
\end{proof}

As the processing time per edge of \Cref{alg:OptSpaceAlg} is clearly $O(1)$, we obtain the following.
\begin{thm}
	For any $\eps>0$, there exists a semi-streaming algorithm computing a $(2+\eps)$-approximation for MWM, using $O(n \log n \cdot \frac{\log(1/\eps)}{\eps})$ space and $O(1)$ processing time.
\end{thm}

\paragraph{Storing $\bm{\phi_v}$ values} In the paper we implicitly assumed the maximum edge weight $W$ is $\textrm{poly}(n)$, implying the $n$ dual variables $\phi_v$ can be stored using $O(n\log n)$ bits. For general $W$, this space requirement is $\Omega(n \log W)$. We briefly outline how to improve this space usage to $O(\log \log W + n\log n)$ bits by only keeping the \emph{ratio} of these $\phi_v$ and the maximum weight observed so far, $W_{\max}$. First, by rounding all edge weights down to the nearest integer power of $(1+\eps)$ (increasing the approximation ratio by at most a $(1+\eps)$ multiplicative factor in the process), we can store $W_{\max}$ by simply storing $\log_{1+\eps} W$, using $O(\log \log_{1+\eps} W) = O\big(\log \log W + \log(1/\eps))$~bits. 
Next, upper bounding the contribution of edge weights below $\eps W_{\max}/n^2$ to $\phi_v$ by $\eps W_{max}/n$, we can store $\phi_v$ using a bit array representing the $O(\log_{1+\eps} (n^2/\eps)) = O\big(\frac{\log n + \log(1/\eps)}{\eps}\big)$ possible values summed this way by $\phi_v$. (Note that the values summed by $\phi_v$ are all distinct by \Cref{obs:exponentialIncrease} and rounding of weights to powers of $(1+\eps)$.)
Arguments similar to those of Footnote \ref{unbounded-W} 
show that this approach keeps the $(2+O(\eps))$ approximation guarantee, while by the above it only requires $O(\log \log W + n\log n)$ bits of memory for any constant $\eps > 0$. That is, for any $W=2^{2^{O(n \log n)}}$, this is still $O(n \log n)$ bits.

\paragraph{Acknowledgment} Mohsen Ghaffari is grateful to Gregory Schwartzman for sharing a write-up of \cite{PS17-StreamingMatching}, and to Ami Paz and Gregory Schwartzman for feedback on an earlier draft of this article. The work of David Wajc is supported in part by NSF grants CCF-1618280, CCF-1814603, CCF-1527110 and NSF CAREER award CCF-1750808.




\bibliography{ref}

\end{document}